\documentclass[a4paper,UKenglish,cleveref,hyperref,autoref]{lipics-v2021}

\title{Treewidth Inapproximability and Tight ETH Lower Bound}
\titlerunning{Treewidth Inapproximability and Tight ETH Lower Bound}

\author{\'{E}douard Bonnet}{Univ Lyon, CNRS, ENS de Lyon, Université Claude Bernard Lyon 1, LIP UMR5668, France \and \url{http://perso.ens-lyon.fr/edouard.bonnet/}}{edouard.bonnet@ens-lyon.fr}{https://orcid.org/0000-0002-1653-5822}{}

\authorrunning{\'E. Bonnet}

\Copyright{Édouard Bonnet}

\category{}

\relatedversion{}

\supplement{}

\funding{The author was supported by the ANR project TWIN-WIDTH (ANR-21-CE48-0014-01).}

\acknowledgements{We thank Hans Bodlaender, Tuukka Korhonen, Daniel Lokshtanov, Nicolas Trotignon, and Rémi Watrigant for some valuable feedback.}

\nolinenumbers 

\hideLIPIcs  

\EventEditors{John Q. Open and Joan R. Access}
\EventNoEds{2}
\EventLongTitle{42nd Conference on Very Important Topics (CVIT 2016)}
\EventShortTitle{CVIT 2016}
\EventAcronym{CVIT}
\EventYear{2016}
\EventDate{December 24--27, 2016}
\EventLocation{Little Whinging, United Kingdom}
\EventLogo{}
\SeriesVolume{42}
\ArticleNo{23}

\usepackage[utf8]{inputenc}  

\usepackage[T1]{fontenc}
\usepackage{lmodern}
\usepackage[colorinlistoftodos,bordercolor=orange,backgroundcolor=orange!20,linecolor=orange,textsize=normalsize]{todonotes}

\usepackage{amsmath}  
\usepackage{amssymb}

\makeatletter
\newcommand\@@vertwv[3]{\mathop{\ooalign{%
  $#1#3$\cr
  \hfil$#1\shortmid$\hfil\cr
  \hfil\raisebox{#2}{$#1\shortmid$}\hfil\cr}}}
\newcommand\@vertwv[1]{\mathchoice
  {\@@vertwv{\displaystyle     }{0.38ex}{#1}}%
  {\@@vertwv{\textstyle        }{0.38ex}{#1}}%
  {\@@vertwv{\scriptstyle      }{0.28ex}{#1}}%
  {\@@vertwv{\scriptscriptstyle}{0.21ex}{#1}}}
\newcommand\vertwedge{\@vertwv\wedge}
\newcommand\vertvee  {\@vertwv\vee  }
\makeatother

\usepackage{bbm}
\usepackage{accents}
\usepackage{complexity}

\usepackage{booktabs}
\usepackage{bm}
\usepackage{paralist}
\usepackage{fixmath}

\makeatletter
\newtheorem*{rep@theorem}{\rep@title}
\newcommand{\newreptheorem}[2]{%
\newenvironment{rep#1}[1]{%
 \def\rep@title{#2 \ref{##1}}%
 \begin{rep@theorem}}%
 {\end{rep@theorem}}}
\makeatother

\newreptheorem{theorem}{Theorem}
\newreptheorem{lemma}{Lemma}

\usepackage{xspace}
\usepackage{tikz}

\usepackage[ruled,vlined,linesnumbered]{algorithm2e}

\usetikzlibrary{fit}
\usetikzlibrary{arrows}
\usetikzlibrary{patterns}
\usetikzlibrary{calc}
\usetikzlibrary{shapes}
\usetikzlibrary{positioning}
\usetikzlibrary{math}
\usetikzlibrary{shapes.symbols}
\usetikzlibrary{decorations.pathreplacing,calligraphy}
\usepackage[scr=boondox,scrscaled=1.05]{mathalfa}

\newcommand\tw{\text{tw}}

\begin{document}

\maketitle

\begin{abstract}
   Despite the algorithmic importance of treewidth, both its complexity and approximability present large knowledge gaps. 
  While the best currently known polynomial-time approximation algorithm has ratio $O(\sqrt{\log \text{OPT}})$, no approximation factor could be ruled out under P $\neq$ NP alone.
  There are $2^{O(n)}$-time algorithms to compute the treewidth of $n$-vertex graphs, but the Exponential-Time Hypothesis (ETH) was only known to imply that $2^{\Omega(\sqrt n)}$ time is required.
  The reason is that all the known hardness constructions use \textsc{Cutwidth} or \textsc{Pathwidth} on bounded-degree graphs as an intermediate step in a~long chain of reductions, for which neither inapproximability nor sharp ETH lower bound is known.

  We present a simple, self-contained reduction from \textsc{3-SAT} to \textsc{Treewidth}.
  Our reduction partially closes the first gap and fully resolves the second.  
  Namely, we show that 1.00005-approximating \textsc{Treewidth} is NP-hard, and solving \textsc{Treewidth} exactly requires $2^{\Omega(n)}$ time, unless the ETH fails.
  We further derive, under the latter assumption, that there are some constants $\delta > 1$ and~$c>0$ such that $\delta$-approximating \textsc{Treewidth} requires time $2^{\Omega(n/\log^c n)}$.
\end{abstract}

\section{Introduction}\label{sec:intro}

Treewidth~\cite{Bertele73,Halin76,RobertsonS84} and tree-decompositions\footnote{See their definition in~\cref{subsec:def-tw}.} have a~central role in algorithm design and graph theory, among other areas.
While it is known that computing the treewidth of a~graph is NP-complete~\cite{Arnborg87}, various approximation and/or parameterized algorithms, and exact exponential algorithms have been developed.
The current Pareto front features, on $n$-vertex graphs of treewidth~$k$, a~polynomial-time $O(\sqrt{\log k})$-approximation algorithm~\cite{Feige08}, an~$O(1.7347^n)$-time exact algorithm~\cite{FominTV15}, a~$2^{O(k)}n$-time 2-approximation algorithm~\cite{Korhonen21}, a~$2^{O(k^3)}n$-time exact algorithm~\cite{Bodlaender96}, a~$2^{O(k^2)}n^4$-time exact algorithm, and a $2^{O((k \log k)/\varepsilon)}n^4$-time $(1+\varepsilon)$-approximation algorithm for any $\varepsilon > 0$~\cite{KorhonenL23b}; see the latter reference for a~detailed overview of the state of the art.

On the complexity side, no tight lower bound was previously known.
The original reduction~\cite{Arnborg87}, as well as subsequent constructions strengthening the NP-hardness of \textsc{Treewidth} to graphs of maximum degree at most~9~\cite{BodlaenderT97}, and even to cubic graphs~\cite{Bodlaender23}, all rely on the NP-hardness of \textsc{Cutwidth} or \textsc{Pathwidth} on bounded-degree graphs.
However, the known reductions for the latter results (see for instance~\cite{Monien88}) incur a~quadratic blow-up in the input size, and do not preserve any multiplicative\footnote{Some \emph{additive} inapproximability is known for treewidth~\cite{BodlaenderGHK95}.} inapproximability.
Therefore, prior to the current paper, a~polynomial-time approximation scheme (PTAS) for \textsc{Treewidth} could not be ruled out under the sole assumption\footnote{Assuming the so-called \emph{Small Subset Expansion Hypothesis} (SSEH)---that the edge expansion of sublinear vertex subsets is hard to approximate---it can be shown that any constant-factor approximation of \textsc{Treewidth} is NP-complete~\cite{Wu14}. However, SSEH is a~stronger assumption than the \emph{Unique Games Conjecture}, whose truth does not gather any consensus; let alone one comparable to that of P $\neq$ NP.} that P $\neq$ NP, and the best lower bound based on the Exponential-Time Hypothesis\footnote{The ETH asserts that there is a~$\lambda > 1$ such that no algorithm solves $n$-variable \textsc{3-SAT} in $O(\lambda^n)$ time~\cite{Impagliazzo01}.} (ETH) only implied that \textsc{Treewidth} requires~$2^{\Omega(\sqrt n)}$ time (see for instance the appendix of~\cite{KorhonenL23b}).

In this paper, we present a~simple self-contained linear reduction from \textsc{3-SAT} to \textsc{Treewidth}.
This improves our understanding in the approximability and complexity of treewidth.
We make the reduction modular so that depending on the instances of \textsc{3-SAT} we start with, we derive the following three theorems.

\begin{theorem}\label{thm:no-ptas}
  1.00005-approximating \textsc{Treewidth} is NP-hard. 
\end{theorem}

\begin{theorem}\label{thm:eth-lb}
  Unless the ETH fails, \textsc{Treewidth} requires $2^{\Omega(n)}$ time on $n$-vertex graphs.
\end{theorem}

\begin{theorem}\label{thm:eth-inapprox}
  Unless the ETH fails, there are~constants $\delta>1$ and $c>0$ such that $\delta$-approximating \textsc{Treewidth} requires $2^{\Omega(n/\log^c n)}$ time on $n$-vertex graphs.
\end{theorem}

The first theorem rules out a~PTAS for \textsc{Treewidth} unless P $=$ NP (and establishes its APX-hardness).
The polynomial-time approximability of \textsc{Treewidth} has often been raised as an important open question; see for instance \cite{BodlaenderGHK95,Feige08,Amir10,Kintali10,FominV12,Kozawa14,Wu14,Yamazaki18,BelbasiF22}.
There is still a~very large gap between this inapproximability factor and the $O(\sqrt{\log \text{OPT}})$-approximation algorithm of Feige, Hajiaghayi, and Lee~\cite{Feige08}, but at least the lower bound has moved for the first time.
The second theorem establishes a~tight ETH lower bound: running time $2^{O(n)}$ (see for instance~\cite{FominTV15}) is best possible to exactly compute treewidth, unless the ETH fails.
This was posted as an open problem in~\cite{KorhonenL23b} and \cite[Question 3]{KorhonenThesis}.
The third theorem combines the previous two hardness features.
In particular, it (loosely) complements the $2^{O((k \log k)/\varepsilon)}n^{O(1)}$-time $(1+\varepsilon)$-approximation algorithm of Korhonen and Lokshtanov~\cite{KorhonenL23b}, in the sense that even for some fixed sufficiently small $\varepsilon>0$ (depending on the ETH constant~$\lambda$), there is a~constant~$c>0$ such that the exponent $O(k \log k)$ in the running time cannot be improved to $O(k/\log^c k)$.

\medskip

\textbf{Techniques.}
The known hardness constructions leverage the fact that on co-bipartite graphs (complements of bipartite graphs), tree-decompositions behave orderly.
They are tame enough to allow a~simple reduction from \textsc{Cutwidth},\footnote{We will \emph{not} need the definition of \textsc{Cutwidth}.} and actually treewidth and pathwidth are equal on co-bipartite graphs \cite{Arnborg87}.
However, it seems challenging to design a~linear reduction from \textsc{3-SAT} to \textsc{Treewidth} on these graphs.
It is indeed unclear how to design a~(binary) choice gadget in this restricted setting.

We thus move to co-tripartite graphs $G$ with tripartition $(A,B,C)$ (into cliques) where $A$ encodes the clauses, and $B \cup C$ encodes the variables, with $B$ associated to their positive form, and $C$, their negation.
More precisely, a~blow-up (where vertices are replaced by clique modules) of a~semi-induced matching\footnote{Here, an induced matching after removing the edges of the cliques $B$ and $C$.} between $B$ and~$C$ constitutes our variable gadgets.
For $A$, we add $7=2^3-1$ vertices for each clause, one for each partial satisfying assignment of the clause, each adjacent to the three modules corresponding to its literals; see~\cref{fig:construction}.
This is to turn the disjunctive nature of a~3-clause into some conjunctive encoding, which better fits tree-decompositions.

Although co-tripartite graphs provide the greater generality (compared to co-bipartite graphs) that allows us to simply design choice gadgets, their tree-decompositions are still tame enough.
In particular, they always admit a~tree-decomposition of minimum width whose underlying tree is a~subdivided claw, and whose three leaf bags contain $A$, $B$, and $C$, respectively.
The crucial property that these tree-decompositions share is that the bag of the unique degree-3 node is a~vertex cover of the graph $I(G)$, defined as $G$ deprived of the edges within the cliques $A$, $B$, and~$C$. 
A~reader familiar with the notion of \emph{bramble} may already observe that the edge set of $I(G)$ is a~bramble of $G$, hence the vertex cover number of $I(G)$ minus one indeed lower bounds the treewidth of~$G$; a~useful fact for the unsatisfiable \textsc{3-SAT} instances.

This is where moving from co-bipartite graphs to co-tripartite graphs is decisive: \textsc{Min Vertex Cover} is polynomial-time solvable on bipartite graphs, but NP-complete on tripartite graphs.
We can thus simply rely on the hardness of~\textsc{Min Vertex Cover}, while this was impossible for the previous reductions based on co-bipartite graphs.
This is indeed what we do.
The treewidth upper bound for the satisfiable \textsc{3-SAT} instances can easily be derived via the \emph{Cops and Robber game}.
We will in fact \emph{not} use \emph{brambles} nor the \emph{Cops and Robber game} in order to make the paper self-contained and more widely accessible.
Readers comfortable with these notions will be able to skip~\cref{subsec:co-trip}, and find their own alternative proof in~\cref{subsec:forward}.

\medskip

\textbf{Further work.}
As far as the polynomial-time approximability of \textsc{Treewidth} is concerned, there is still a~large gap between \cref{thm:no-ptas} and the $O(\sqrt{\log k})$-approximation algorithm \cite{Feige08}.
As we will observe in~\cref{subsec:instantiate}, the inapproximability factor of 1.00005 could be somewhat improved by reducing from \textsc{Max 2-SAT} instead of \textsc{Max 3-SAT}.
However, it would then remain in the regime of $1+\alpha$ with $\alpha \ll 1$. 
Furthermore, there is a~straightforward 3-approximation algorithm in co-tripartite instances, and more generally a~constant-approximation algorithm in instances whose vertex sets can be covered by a~constant number of cliques.
If indeed \textsc{Treewidth} cannot be polynomial-time approximated within any constant ratio (unless P~$=$~NP), showing such a~statement will require new ideas.
As a~more humble intermediate step, pushing the inapproximability factor to~2 (or, at~least, $2-\varepsilon$ for any $\varepsilon > 0$) would nicely contrast with the~$2^{O(k)}n$-time 2-approximation algorithm~\cite{Korhonen21}.

Now we know that, for any constant $c$, a~$2^{o(k)}n^c$-time exact algorithm for \textsc{Treewidth} would refute the ETH (by~\cref{thm:eth-lb}), a~natural remaining open question is to close the gap with the $2^{O(k^2)}n^4$-time exact algorithm~\cite{KorhonenL23b}.
We finally observe that the pathwidth of the constructed hard instances is in general different from their treewidth.
We nonetheless wonder if our ideas could help establish some inapproximability results for \textsc{Pathwidth} or other related parameters. 

\section{Definitions and Notation}

If $i$ is a~positive integer, we denote by $[i]$ the set of integers $\{1,2,$ $\ldots,i\}$.

\subsection{Classical Graph-Theoretic Notions and Notation}

We denote by $V(G)$ and $E(G)$ the set of vertices and edges of a graph $G$, respectively.
For $S \subseteq V(G)$, the \emph{subgraph of $G$ induced by $S$}, denoted $G[S]$, is obtained by removing from $G$ all the vertices that are not in $S$ (together with their incident edges).
A~set $X \subseteq V(G)$ is \emph{connected} (in $G$) if $G[X]$ has a~single connected component, and \emph{disconnected} otherwise.
A~graph is \emph{co-tripartite} (resp.~\emph{co-bipartite}) if it is the complement of a~tripartite graph (resp.~bipartite) graph, or equivalently can have its vertex set partitioned into three (resp.~two) cliques.
A~\emph{vertex cover} of a~graph $G$ is a~subset $S \subseteq V(G)$ such that every edge of~$G$ has at~least one of its two endpoints in~$S$.

We denote by $N_G(v)$ and $N_G[v]$, the open, respectively closed, neighborhood of $v$ in $G$.
We define $N_G(S) := \bigcup_{v \in S}N_G(v) \setminus S$ for $S \subseteq V(G)$, and $N_G[S] := N_G(S) \cup S$.
A~\emph{module} in $G$ is a~subset $Y \subseteq V(G)$ such that for every $u, v \in Y$, $N_G(u) \setminus Y = N_G(v) \setminus Y$.
The \emph{degree} $d_G(v)$ of a~vertex $v \in V(G)$ is the size of $N_G(v)$, and the \emph{maximum degree} of $G$ is $\max_{v \in V(G)} d_G(v)$.
In all the previous notations, we may omit the graph subscript if it is clear from the context.
A~\emph{subdivided claw} is any tree with exactly three leaves.
Note that any subdivided claw has exactly one vertex of degree~3, and apart from it and its three leaves, only vertices of degree~2.

\subsection{Tree-Decompositions and Treewidth}\label{subsec:def-tw}

A~\emph{tree-decomposition} of a~graph $G$ is a~pair $(T,\beta)$ where $T$ is a~tree and $\beta$ is a~map from $V(T)$ to $2^{V(G)}$ satisfying the following properties:
\begin{compactitem}
\item for every $v \in V(G)$, $\{t \in V(T)~:~v \in \beta(t)\}$ induces a~non-empty subtree of $T$, and
\item for every $uv \in E(G)$, there is a~$t \in V(T)$ such that $\{u,v\} \subseteq \beta(t)$.
\end{compactitem}
The \emph{width} of $(T,\beta)$ is defined as $\max_{t \in V(T)} |\beta(t)| - 1$, and the \emph{treewidth} of $G$, denoted by $\tw(G)$, is the minimum width of $(T,\beta)$ taken among every tree-decomposition $(T,\beta)$ of~$G$.
A~\emph{path-decomposition} is a~tree-decomposition $(T,\beta)$ where $T$ is a~path.
And \emph{pathwidth} is defined as \emph{treewidth} with \emph{path-decompositions} instead of \emph{tree-decompositions}. 

We may call $\beta(t)$ the \emph{bag} of $t \in V(T)$.
We also call \emph{trace} of $v \in V(G)$ the set $\{t~\in~V(T)~: v \in \beta(t)\}$.
We may say that an edge $uv \in E(G)$ is \emph{covered} by a~tree-decomposition $(T,\beta)$ (not a~priori claimed to be one of~$G$) if there is a~$t \in V(T)$ such that $\{u,v\} \subseteq \beta(t)$.
More specifically, the edge $uv$ is \emph{covered} by node $t$.
A~pair $(T,\beta)$ where $T$ is a~tree and $\beta$ is a~map from the nodes of~$T$ to subsets of some universe $U$ is a~(valid) tree-decomposition (of some graph) if the trace of each element of $U$ induces a~non-empty subtree of~$T$.
It further is a~tree-decomposition \emph{of $G$} if $U=V(G)$ and every edge of $G$ is covered by~$(T,\beta)$.

\section{Treewidth Hardness}\label{sec:main}

Let $\varphi$ be a~3-CNF formula.
We denote by $x_1, \ldots, x_n$ the variables of $\varphi$, and by $c_1, \ldots, c_m$, its clauses, each of them on exactly three literals.

\subsection{Construction of \boldmath{$G(\varphi)$}}\label{subsec:constr}

We build a~co-tripartite graph $G := G(\varphi)$ with $2 \gamma n+7m$ vertices, where $\gamma$ is a~natural number to be instantiated.
The set $V(G)$ is partitioned into $(A,B,C)$ with $G[A]$, $G[B]$, and $G[C]$ each being a~clique.
The set $A$ represents the clauses, and $B \cup C$, the variables, with $B$ corresponding to their positive form, and $C$ their negation.

For every variable $x_i$, we add $\gamma$ vertices\footnote{Only to obtain \cref{thm:eth-inapprox}, $\gamma$ will depend on $i$, and will be denoted $\gamma_i$.} $b^1_i, \ldots, b^\gamma_i$ to $B$, and $\gamma$ vertices $c^1_i, \ldots, c^\gamma_i$ to $C$.
We keep the value of the natural number $\gamma$ generic.
The only constraint on $\gamma$ is that there is \emph{no literal} with $p$ positive occurrences and $q$ negative occurrences in $\varphi$ such that $4p+3q>\gamma$.
To clarify, every clause of $\varphi$ containing $\neg x_i$ counts for a~\emph{positive} occurrence of literal $\neg x_i$, and every clause of $\varphi$ containing $x_i$ counts for a~negative occurrence of literal $\neg x_i$.

(For the concreteness of showing~\cref{thm:no-ptas}, the reader can assume that every variable appears exactly twice positively and exactly twice negatively, and $\gamma := 4 \cdot 2 + 3 \cdot 2 = 14$.) 
We set $B(x_i) := \{b^1_i, \ldots, b^\gamma_i\}$ and $C(x_i) := \{c^1_i, \ldots, c^\gamma_i\}$.
For every $i,i' \in [n]$ and $h,h' \in [\gamma]$, the two vertices $b^h_i, c^{h'}_{i'}$ are made adjacent if and only if $i = i'$.
Each set $B(x_i)$ and $C(x_i)$ will remain a~module in~$G(\varphi)$, and we will therefore seldom refer to~$b^h_i$ individually.

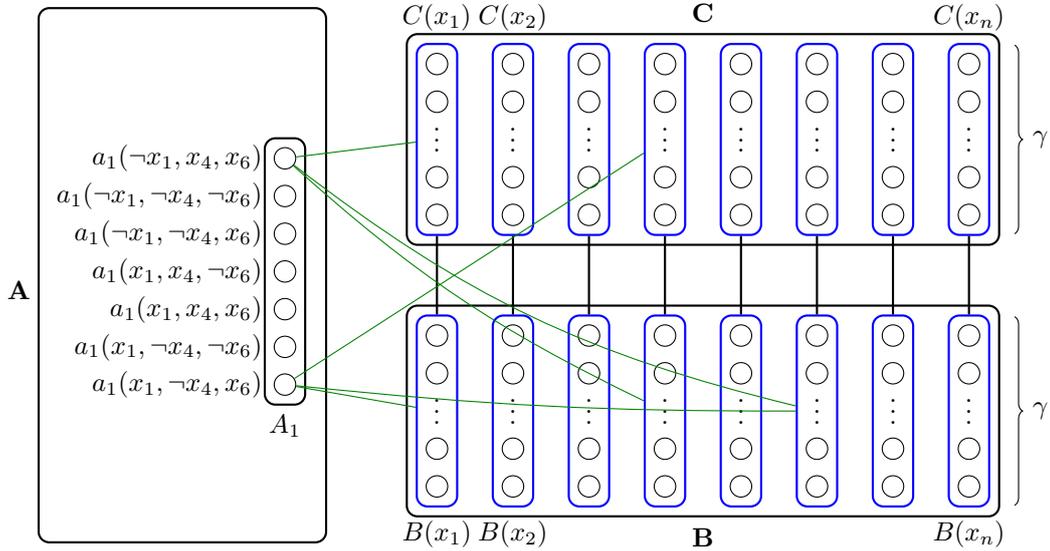
\begin{figure}[h!]
  \phantomsection
  \centering
  \begin{tikzpicture}[vertex/.style={draw,circle,inner sep=0.1cm}]
    \def\n{8}
    \def\hs{1}
    \def\vs{0.5}
    \def\oc{3.6}
    
    \foreach \i in {1,...,\n}{
      \foreach \h in {1,2,4,5}{
        \node[vertex] (b\i\h) at (\i * \hs - \hs, \vs * \h) {} ;
        \node[vertex] (c\i\h) at (\i * \hs - \hs, \vs * \h + \oc) {} ;
      }
      \node at (\i * \hs - \hs, 3.2 * \vs) {$\vdots$} ;
      \node at (\i * \hs - \hs, 3.2 * \vs + \oc) {$\vdots$} ;
      \node[draw,blue,rounded corners,thick,fit=(b\i1) (b\i5)] (B\i) {} ;
      \node[draw,blue,rounded corners,thick,fit=(c\i1) (c\i5)] (C\i) {} ;
      \draw[thick] (B\i) -- (C\i) ;
    }

    \node[draw,rounded corners,thick,fit=(B1) (B\n)] (B) {} ;
    \node[draw,rounded corners,thick,fit=(C1) (C\n)] (C) {} ;

    \foreach \i/\il in {1/1,2/2,\n/n}{
      \node at (\i * \hs - \hs, -0.29 * \vs) {$B(x_{\il})$} ;
      \node at (\i * \hs - \hs, 6.29 * \vs + \oc) {$C(x_{\il})$} ;
    }

    \node at (4.5 * \hs - \hs, -0.35 * \vs) {$\mathbf{B}$} ;
    \node at (4.5 * \hs - \hs, 6.35 * \vs + \oc) {$\mathbf{C}$} ;

    \draw [decorate,decoration={brace,amplitude=3pt, mirror}] (\n * \hs - 0.4 * \hs, 0.5 * \vs) to node [midway,right] {$~\gamma$} (\n * \hs - 0.4 * \hs, 5.5 * \vs) ;
    \draw [decorate,decoration={brace,amplitude=3pt, mirror}] (\n * \hs - 0.4 * \hs, 0.5 * \vs + \oc) to node [midway,right] {$~\gamma$} (\n * \hs - 0.4 * \hs, 5.5 * \vs + \oc) ;

    \foreach \h in {1,...,7}{
      \node[vertex] (s\h) at (-2 * \hs, \h * \vs + 2.7 * \vs) {} ;
    }
    \foreach \h/\l in {1/{$a_1(x_1,\neg x_4,x_6)~~$},2/{$a_1(x_1,\neg x_4,\neg x_6)~~~~$},3/{$a_1(x_1,x_4,x_6)$},4/{$a_1(x_1,x_4,\neg x_6)~~$},5/{$a_1(\neg x_1,\neg x_4,x_6)~~~~$},6/{$a_1(\neg x_1,\neg x_4,\neg x_6)~~~~~~$},7/{$a_1(\neg x_1,x_4,x_6)~~$}}{
      \node at (-3.38 * \hs, \h * \vs + 2.7 * \vs) {\l} ;
    }
    \node at (-2 * \hs, 2.6 * \vs) {$A_1$} ;
    \node[draw,rounded corners,thick,fit=(s1) (s7)] (A1) {} ;

    \node at (-5 * \hs, 0) (bl) {} ;
    \node at (-1.7 * \hs, 6 * \vs + \oc) (tr) {} ;
    \node[draw,rounded corners,thick,fit=(bl) (tr)] (A) {} ;

    \node at (-5.5 * \hs, 6.2 * \vs) {$\mathbf{A}$} ;

    \foreach \i/\j/\b in {1/B1/0,1/C4/0,1/B6/-3, 7/C1/0,7/B4/-8,7/B6/-12}{
      \draw[green!50!black] (s\i) to [bend left = \b] (\j) ;
    }
  \end{tikzpicture}
  \caption{Illustration of $G(\varphi)$ with $n=8$.
    For the sake of not cluttering the picture, we did not draw the edges within the cliques $A$, $B$, and $C$.
    We also only represented the vertices of $A$ encoding the clause $c_1=x_1 \lor \neg x_4 \lor x_6$, i.e., $A_1$, and only drew the edges incident to $a_1(x_1,\neg x_4,x_6)$ and to~$a_1(\neg x_1,x_4,x_6)$.
    The blue vertical boxes are modules.
    Hence an edge incident to a~blue box represents all the adjacencies between the enclosed vertices and the vertex (or vertices of another blue box) at the other end of the edge.}   
  \label{fig:construction}
\end{figure}

For every clause $c_j = \ell_1 \lor \ell_2 \lor \ell_3$ with $\ell_p = (\neg) x_{i_p}$ for $p \in [3]$, we add a~set $A_j$ of 7 vertices to $A$, one for each assignment of its three variables satisfying~$c_j$.
For $(s_1, s_2, s_3) \in (\{\ell_1, \neg \ell_1\} \times \{\ell_2, \neg \ell_2\} \times \{\ell_3, \neg \ell_3\}) \setminus \{(\neg \ell_1, \neg \ell_2, \neg \ell_3)\}$, we denote by $a_j(s_1,s_2,s_3)$ the corresponding vertex of $A$, while syntactically replacing $\neg \neg x$ by $x$.
For each $p \in [3]$, we make $a_j(s_1,s_2,s_3)$ fully adjacent to $B(x_{i_p})$ if $s_p = x_{i_p}$, or to $C(x_{i_p})$ if $s_p = \neg x_{i_p}$.
This finishes the construction of~$G$; see~\cref{fig:construction}.

\subsection{Tree-Decompositions of Co-tripartite Graphs}\label{subsec:co-trip}

Here we show that every co-tripartite graph $G$ with tripartition $(A,B,C)$ has a~tree-decomposition $(T,\beta)$ of width $\tw(G)$ such that $T$ is a~subdivided claw, and the bag of its vertex of degree~3 is a~vertex cover of \emph{$I(G)$}, the graph $G$ deprived of the edges within the cliques $A$, $B$, and $C$.
A~reader familiar with \emph{brambles} may observe that the family of pairs $\mathcal B := \{\{u,v\}~:~uv \in E(G) \setminus \bigcup_{X \in \{A,B,C\}} E(G[X]) = E(I(G))\}$ is a~bramble of~$G$, and that a~hitting set of $\mathcal B$ is, by definition, a~vertex cover of $I(G)$; thereby reaching the desired treewidth lower bound.
They may then skip \cref{subsec:co-trip}.
We chose this presentation, as \cref{lem:neat} better prepares to~\cref{subsec:forward,subsec:backward}.

We start by recalling a~classical lemma.

\begin{lemma}\label{lem:clique-in-tree-dec}
  Let $G$ be a~graph, $X$ be a~clique of $G$, and $(T,\beta)$ be a~tree-decomposition of~$G$.
  Then there is a~node $t \in V(T)$ such that $X \subseteq \beta(t)$.
\end{lemma}
\begin{proof}
  Every two vertices of $X$ have to appear together in some bag of~$(T,\beta)$.
  Let $v_1, \ldots, v_{|X|}$ enumerate $X$, and assume for the sake of contradiction, that the traces of $v_1, \ldots, v_s$ (with $s < |X|$) have a~non-empty common intersection $Y$ such that the trace of $v_{s+1}$ is disjoint from~$Y$.
  Note that $T[Y]$ is a~tree, and the trace of $v_{s+1}$ is included in the vertex set $Z$ of a~connected component of $T[V(T) \setminus Y]$.
  There is a~single node $z \in Z$ adjacent to a~node of~$Y$.
  By definition of~$Y$, there is a~$v_i$ with $i \in [s]$ such that $v_i \notin \beta(z)$.
  As traces should be connected, this contradicts that there is a~bag of~$(T,\beta)$ containing both $v_i$ and $v_{s+1}$.
\end{proof}

\Cref{lem:clique-in-tree-dec} justifies the existence of $t_A, t_B, t_C$ in the assumption of the following lemma.

\begin{lemma}\label{lem:Y}
  Let $G$ be a~co-tripartite graph with tripartition $(A,B,C)$, and $(T,\beta)$ be a~tree-decompo\-si\-tion of~$G$.
  Let $t_A, t_B, t_C \in V(T)$ be such that $A \subseteq \beta(t_A)$, $B \subseteq \beta(t_B)$, and $C \subseteq \beta(t_C)$.
  Let $T'$ be the minimal subtree of $T$ containing $t_A, t_B, t_C$.
  Then $(T',\beta)$ is a~tree-decomposition of $G$ (where $\beta$ is used, for simplicity's sake, for its restriction to $V(T')$).
\end{lemma}

\begin{proof}
  In a~tree, the intersection of two subtrees is also a~subtree.
  Thus $(T',\beta)$ is a~valid tree-decompo\-si\-tion (of some graph).
  We show that $(T',\beta)$ is a~tree-decomposition of $G$.
  As $A \subseteq \beta(t_A)$, $B \subseteq \beta(t_B)$, and $C \subseteq \beta(t_C)$, every vertex of $G$ is in some bag of~$(T',\beta)$.
  We next argue that no node of $V(T) \setminus V(T')$ is actually useful to cover the edges of~$G$.
  Let $t \in V(T) \setminus V(T')$ and $t'$ be the node of $T'$ in the shortest path from $t$ to $V(T')$.
  It holds that $\beta(t) \subseteq \beta(t')$, as otherwise the trace of any~vertex of $\beta(t) \setminus \beta(t')$ in $(T,\beta)$ is disconnected.
  Hence $t$ does not cover edges that are not already covered by~$t'$.
\end{proof}

If $G$ is a~co-tripartite graph with tripartition $(A,B,C)$, we recall that $I(G)$ is the graph obtained from $G$ by removing all the edges within the cliques $A$, $B$, and $C$.

\begin{lemma}\label{lem:neat}
  Let $G$ be a~co-tripartite graph with tripartition $(A,B,C)$.
  Then $G$ has a~tree-decompo\-si\-tion $(T,\beta)$ of width $\tw(G)$ such that
  \begin{compactitem}
  \item $T$ is a~subdivided claw,
  \item the three leaves $t_A$, $t_B$, $t_C$ of $T$ satisfy $X \subseteq \beta(t_X)$ for every $X \in \{A,B,C\}$, and
  \item the unique degree-3 node $t_{\vertwedge}$ of $T$ is such that $\beta(t_{\vertwedge})$ is a~vertex cover of $I(G)$.
  \end{compactitem}
\end{lemma}

\begin{proof}
  By applying~\cref{lem:Y} with a~tree-decompo\-si\-tion of minimum width, $G$ has a~tree-decompo\-si\-tion $(T,\beta)$ of width $\tw(G)$ such that $T$ is a~subdivided claw whose three leaves $t_A$, $t_B$, $t_C$ verify $X \subseteq \beta(t_X)$ for every $X \in \{A,B,C\}$.
  (Indeed, if the minimal subtree connecting $t_A$, $t_B$, $t_C$ is a~path, one can simply add a~neighbor with the same bag to whichever of $t_A$, $t_B$, $t_C$ is not yet a~distinct leaf node.
  Note that this process adds up to three new bags.)
  Hence the first two items of the lemma are satisfied.
  
  Let $t_{\vertwedge}$ be the unique node of $T$ with degree~3.
  Assume for the sake of contradiction that an edge $uv \in E(I(G))$ is such that $\{u,v\} \cap \beta(t_{\vertwedge}) = \emptyset$.
  Without loss of generality, let us assume that $u \in A$ and $v \in B$.
  We claim that no bag of $(T,\beta)$ may then include $\{u,v\}$.
  Indeed, vertex $u$ may only be present in the path of $T$ going from $t_A$ to $t_{\vertwedge}$, excluding $t_{\vertwedge}$ itself, while $v$ is only present within the path from $t_B$ to $t_{\vertwedge}$ (excluded).
\end{proof}

\subsection{If one can satisfy \boldmath{$m'$} clauses, \boldmath{$\tw(G) \leqslant \gamma n + 7m - m' +\gamma-1$}}\label{subsec:forward}

We go back to the particular co-tripartite graph $G := G(\varphi)$ with tripartition $(A,B,C)$ constructed in \cref{subsec:constr}, and exhibit, when a~truth assignment satisfies $m'$ clauses of~$\varphi$, a~tree-decomposition $(T,\beta)$ of~$G$ where every bag has size at~most $\gamma(n+1)+7m-m'$.
Let us recall that $|A|=7m$ and $|B|=|C|=\gamma n$.
The tree $T$ is a~subdivided claw.
Let $t_{\vertwedge}$ be its degree-3 vertex, and $t_A, t_B, t_C$ be its three leaves.

\medskip

\textbf{The bag of $\bm{t_{\vertwedge}}$.}
Let $\mathcal A$ be a~truth assignment of $x_1, \ldots, x_n$ satisfying $m' \leqslant m$ clauses of~$\varphi$.
We denote by $\mathcal A^+ \subseteq [n]$ (resp.~$\mathcal A^- \subseteq [n]$) the set of indices $i$ such that $\mathcal A$ sets $x_i$ to true (resp.~to false).
Thus $(\mathcal A^+,\mathcal A^-)$ partitions~$[n]$.
Let us define the bag of~$t_{\vertwedge}$.
We set \[B' := \bigcup_{i \in \mathcal A^+} B(x_i)\text{,~~and~~} C' := \bigcup_{i \in \mathcal A^-} C(x_i).\]
We define the subset $A' \subset A$ of size $7m-m'$, starting from $A$ and removing, for each satisfied clause~$c_j$, the vertex $a_j(s_1,s_2,s_3) \in A_j$ such that $s_1, s_2, s_3$ are all satisfied by $\mathcal A$.
Thus $|A'|=|A|-m'=7m-m'$.
We set $\beta(t_{\vertwedge}) := A' \cup B' \cup C'$.

\medskip

\textbf{Path-decompositions from $\bm{t_{\vertwedge}}$ to the three leaf nodes.}
We now give the path-decompo\-si\-tions from $t_{\vertwedge}$ to $t_X$ for each $X \in \{A,B,C\}$, such that: $X \subseteq \beta(t_X)$, and for every node $t \in V(T)$ in the path from $t_{\vertwedge}$ to $t_X$, for each $Y \in \{A,B,C\} \setminus \{X\}$, $\beta(t) \cap Y \subseteq Y'$.
We will later see that path-decompositions satisfying the previous conditions combine to define a~valid tree-decomposition.

The path-decomposition from $t_{\vertwedge}$ to $t_B$ goes as follows.
Initially, the \emph{active node} is~$t_{\vertwedge}$.
For every index $i \in \mathcal A^-$ in increasing order, add a~neighbor $t'$ to the current active node $t \in V(T)$, and set $\beta(t') := \beta(t) \cup B(x_i)$.
Then add a~neighbor $t''$ to $t'$, and set $\beta(t'') := \beta(t') \setminus C(x_i)$.
Node $t''$ then becomes the active node.
The active node after the last iteration in $\mathcal A^-$ is $t_B$.
Note that $\beta(t_B)=A' \cup B$.
The path-decomposition from $t_{\vertwedge}$ to $t_C$ is defined analogously by replacing $\mathcal A^-$ with $\mathcal A^+$, and swapping the roles of $B(x_i)$ and $C(x_i)$.
In particular, $\beta(t_C)=A' \cup C$.

We finally describe the path-decomposition from $t_{\vertwedge}$ to $t_A$.
The path of $T$ from $t_{\vertwedge}$ to $t_A$ is: $t_{\vertwedge}=t'_0, t_1, t'_1, t_2, t'_2, \ldots, t_n, t'_n=t_A$.
For every integer $i$ from 1 to $n$, we set $\beta(t_i):=\beta(t'_{i-1}) \cup Z_i$ where $Z_i$ is the set of neighbors in $A$ of vertices in $B(x_i)$ if $i \in \mathcal A^+$ (equivalently, if $B(x_i) \subseteq B' \cup C'$), or of vertices in $C(x_i)$ if $i \in \mathcal A^-$ (equivalently, if $C(x_i) \subseteq B' \cup C'$).
And we set $\beta(t'_i):=\beta(t_i) \setminus (B(x_i) \cup C(x_i))$.
Note that $\beta(t_A)=A$.

This finishes the construction of~$(T,\beta)$.
We have three properties to check: the trace of every vertex of~$G$ makes a~non-empty tree in $(T,\beta)$, all edges of $G$ are covered by $(T,\beta)$, and every bag of $(T,\beta)$ has size at~most $\gamma(n+1)+7m-m'$.

\medskip

\textbf{The trace of every $\bm{v \in V(G)}$ induces a non-empty tree.}
The simplest case is when $v \in A'$, as $A'$ is included in \emph{every} bag of $(T,\beta)$.
For each $X \in \{A,B,C\}$, the trace (in $(T,\beta)$) of any $v \in X \setminus X'$ is a~path ending at~$t_X$, since $v \in X \subseteq \beta(t_X)$, and in the path of $T$ going from $t_{\vertwedge}$ to $t_X$, vertex $v$ is added to the bag of some node~$t$, and never removed in the path from $t$ to~$t_X$.
Furthermore, if $\{Y,Z\}=\{A,B,C\} \setminus \{X\}$, no vertex of $X \setminus X'$ is part of a~bag in the path of $T$ from $t_Y$ to $t_Z$.

It remains to check that the traces of vertices of $B' \cup C'$ induce subtrees of~$T$.
This is indeed the case since every vertex $v \in B'$ (resp.~$v \in C'$) is in all the bags along the path from $t_{\vertwedge}$ to~$t_B$ (resp. to~$t_C$), while no vertex in~$B$ (resp. in~$C$) is ever \emph{added} to a~bag in the paths from $t_{\vertwedge}$ to $t_A$ and from $t_{\vertwedge}$ to $t_C$ (resp.~to $t_B$).

\medskip

\textbf{Every edge $\bm{e \in E(G)}$ is covered by $\bm{(T,\beta)}$.}
For each $X \in \{A,B,C\}$, every edge within the clique $X$ is covered by $t_X$, as $X \subseteq \beta(t_X)$.
Therefore, we shall just argue that every edge of $E(I(G))$ is covered by~$(T,\beta)$.

We start with the edges between $B$ and $C$.
Recall that all these edges are of the form $b_i^hc_i^{h'}$ with $i \in [n]$ and $h,h' \in [\gamma]$.
If $i \in \mathcal A^-$, there is, in the path of $T$ from $t_{\vertwedge}$ to $t_B$, a~node where we just added $B(x_i)$, while $C(x_i) \subseteq C'$ is still present (and $C(x_i)$ is removed in the following bag).
Thus this node covers $b_i^hc_i^{h'}$, as well as every edge between $B(x_i)$ and $C(x_i)$.
Analogously, if $i \in \mathcal A^+$, there is in the path of $T$ from $t_{\vertwedge}$ to $t_C$, a~node where we just added $C(x_i)$, while $B(x_i) \subseteq B'$ is still present.

We now consider edges between $A$ and $B \cup C$.
Note that every edge between $A'$ and $B$ (resp.~$C$) is covered by $t_B$ (resp.~$t_C$) since $\beta(t_B)=A' \cup B$ (resp.~$\beta(t_C)=A' \cup C$).
We thus turn our attention to edges between $A \setminus A'$ and $B \cup C$.
By construction, no vertex of $A \setminus A'$ has a~neighbor in $(B \cup C) \setminus (B' \cup C')$.
Hence we are only left with edges between $A \setminus A'$ and $B' \cup C'$.

Fix an edge $uv \in E(G)$ with $u \in A \setminus A'$ and $v \in B' \cup C'$, and let $i \in [n]$ be such that $v$ belongs to~$\beta(t_i)$ but not to~$\beta(t'_i)$.
Note that $i$ is well-defined and corresponds to the index such that $v \in (B(x_i) \cup C(x_i)) \cap (B' \cup C')$.
As $N_G((B(x_i) \cup C(x_i)) \cap (B' \cup C')) \cap A = N_G(v) \cap A \subseteq \beta(t_i)$, edge $uv$ is covered by~$t_i$ (since $u \in N_G(v) \cap A$).

\medskip

\textbf{Every bag of $\bm{(T,\beta)}$ has size at~most $\bm{\gamma(n+1)+7m-m'}$.}
First note that $|\beta(t_{\vertwedge})|=\gamma n+7m-m'$ and that each bag along the path of $T$ between $t_B$ and $t_C$ has either size $\gamma n+7m-m'$ or size $\gamma(n+1)+7m-m'$.
We finally claim that for every $i \in [n]$, $|\beta(t_i)| \leqslant |\beta(t'_{i-1})|+\gamma$ and $|\beta(t'_i)| \leqslant |\beta(t'_{i-1})|$.
For the former inequality, notice that vertices of $B(x_i)$ (resp.~$C(x_i)$) have the same four neighbors within each $A_j$ such that $x_i$ appears positively (resp.~negatively) in $c_j$, and the same three neighbors within each $A_j$ such that $x_i$ appears negatively (resp.~positively) in $c_j$, and no other neighbor in~$A$.
Further recall that every \emph{literal} with $p$ positive occurrences and $q$ negative occurrences in $\varphi$ satisfies $\gamma \geqslant 4p+3q$.
Consequently, to check that $|\beta(t'_i)| \leqslant |\beta(t'_{i-1})|$, simply recall that the bag of $t'_i$ is that of $t_i$ deprived of $\gamma$ vertices.
Therefore every bag along the path from $t_{\vertwedge}$ to $t_A$ has size at~most $|\beta(t'_0)|+\gamma=|\beta(t_{\vertwedge})|+\gamma=\gamma(n+1)+7m-m'$.

\medskip

We have thus established that the treewidth of~$G$ is at~most~$\gamma n + 7m - m'+ \gamma - 1$.

\subsection{If at~most \boldmath{$m''$} clauses are satisfiable, \boldmath{$\tw(G) \geqslant \gamma n + 7m-m''-1$}}\label{subsec:backward}

By~\cref{lem:neat}, there is a~tree-decomposition of $G$ of width $\tw(G)$ with a~bag (that of $t_{\vertwedge}$) that is a~vertex cover of $I(G)$.
Therefore, we shall just argue that every vertex cover of $I(G)$ has size at~least $\gamma n + 7m-m''$, when every truth assignment satisfies at~most $m''$ clauses of~$\varphi$.
We fix any inclusion-wise minimal vertex cover $S$ of $I(G)$, and set $A' := A \cap S$, $B' := B \cap S$, $C' := C \cap S$.
We show the following replacement lemma, to obtain a~more suitable and non-larger vertex cover~$S'$ of~$I(G)$.

\begin{lemma}\label{lem:replacement}
  There is a~vertex cover $S'$ of~$I(G)$ such that $S' \cap (B(x_i) \cup C(x_i)) \in \{B(x_i),C(x_i)\}$ for every $i \in [n]$, and $|S'| \leqslant |S|$. 
\end{lemma}
\begin{proof}
  As, for every $i \in [n]$, every vertex of $B(x_i)$ is adjacent to every vertex of $C(x_i)$, it necessarily holds that $B(x_i) \subseteq S$ or $C(x_i) \subseteq S$.
As $B(x_i)$ (resp.~$C(x_i)$) is a~module (in $G$ and in $I(G)$) and $S$ is inclusion-wise minimal, it further holds that $S \cap (B(x_i) \cup C(x_i)) \in \{B(x_i),C(x_i),B(x_i) \cup C(x_i)\}$. 
  
  We initialize the set $S'$ to $S$.
  For every $i \in [n]$ such that $B(x_i) \cup C(x_i) \subseteq S$, we remove $C(x_i)$ from $S'$, and add $N_G(C(x_i)) \cap A$ to $S'$.
  It was already observed that the condition imposed on $\gamma$ implies that $|N_G(C(x_i)) \cap A| \leqslant \gamma$, so the size of $S'$ may only decrease (as $|C(x_i)|=\gamma$).
  The replacement preserves the property of being a~vertex cover of~$I(G)$, since all the neighbors of $C(x_i)$ end up in $S'$ whenever $C(x_i)$ is removed.
\end{proof}

By~\cref{lem:replacement}, we just need to show that $|S'| \geqslant \gamma n + 7m-m''$.
We note that $S' \cap (B \cup C)$ has size $\gamma n$ and defines a~truth assignment $\mathcal A$ of $x_1, \ldots, x_n$: $\mathcal A$ sets $x_i$ to true if $S' \cap (B(x_i) \cup C(x_i))=B(x_i)$, and to false if $S' \cap (B(x_i) \cup C(x_i))=C(x_i)$.
Seeing $S'$ as $(A \setminus \widehat A) \cup (S' \cap (B \cup C))$ for some $\widehat A \subseteq A$, it~holds that $|S'|=7m- |\widehat A|+\gamma n$.
For $S'$ to be a~vertex cover of $I(G)$, a~vertex $v$ of $\widehat A$ has to have all its neighbors within $B \cup C$ included in~$S'$.
This can only happen if $v=a_j(s_1,s_2,s_3)$, $c_j$ is satisfied by $\mathcal A$, and the literals $s_1, s_2, s_3$ are all satisfied by $\mathcal A$.
This implies that $|\widehat A| \leqslant m''$.
Therefore, $|S'| \geqslant \gamma n+7m-m''$.
We conclude that the treewidth of~$G$ is at~least $\gamma n + 7m-m''-1$.

\subsection{Instantiating the Reduction}\label{subsec:instantiate}

We now apply the previous reduction to hard (to approximate) \textsc{3-SAT} instances in order to show~\cref{thm:no-ptas,thm:eth-lb,thm:eth-inapprox}.
Recall that as long as $\gamma=O(1)$, our reduction produces graphs on $O(n+m)$ vertices from $n$-variable $m$-clause \textsc{3-SAT} formulas.
In particular, it is a~polynomial reduction, and linear in $n$ whenever $m=O(n)$.
For the proof of~\cref{thm:eth-inapprox}, we will slightly adapt the reduction with varying values of possibly superconstant $\gamma_i$ (depending on the number of occurrences of~$x_i$), all at~most~polylogarithmic in~$m$.

\medskip

\textbf{Proof of~\cref{thm:no-ptas}.}
Following Berman, Karpinski, and Scott \cite{Berman03}, we call \emph{$(3,2B)$-SAT} formula a~3-CNF formula where every clause has exactly three literals, and every variable appears exactly twice positively, and exactly twice negatively.
The same authors showed that, for any $\varepsilon > 0$, it is NP-complete to distinguish within $m$-clause $(3,2B)$-SAT formulas those for which at~least $(1-\varepsilon)m$ clauses are satisfiable from those for which at~most $(\frac{1015}{1016}+\varepsilon)m$ clauses are satisfiable.

We observe that $m$-clause $(3,2B)$-SAT formulas have $\frac{3m}{4}$ variables.
Substituting $m' := (1-\varepsilon)m$, $m'' := (\frac{1015}{1016}+\varepsilon)m$, $\gamma := 4 \cdot 2 + 3 \cdot 2 = 14$, and $n = \frac{3m}{4}$, we get that it is NP-hard to distinguish graphs  of treewidth at~most $(10.5+6+\varepsilon)m+13$ from graphs of treewidth at~least $(10.5+7-\frac{1015}{1016}-\varepsilon)m-1$.
We conclude as \[\frac{(17.5-\frac{1015}{1016}-\varepsilon)m-1}{(16.5+\varepsilon)m+13} > 1.00005\] for sufficiently small $\varepsilon$ and large~$m$.

We observe that one can increase the inapproximability gap by performing (with the required adjustments) our reduction from \textsc{Max-2-SAT} where every variable has at~most four occurrences.
This problem has a~better known inapproximability factor than \textsc{Max-$(3,2B)$-SAT}~\cite{Berman03}, and this would also allow to decrease the value of~$\gamma$ (to 7). 

\medskip

\textbf{Proof of~\cref{thm:eth-lb}.}
The Sparsification Lemma of Impagliazzo, Paturi, and Zane~\cite{sparsification} implies that, under the ETH, $n$-variable \textsc{3-SAT} requires $2^{\Omega(n)}$ time even within formulas where every variable appears at most $B$ times, for some constant~$B$.
We denote this fragment of \textsc{3-SAT} by \textsc{3-SAT-$B$} (even when $B$ may depend on~$n$ or~$m$).
We set $\gamma := 4B$ so that the condition on $\gamma$ is verified, and we duplicate $\gamma+1$ times the initial hard formula $\varphi'$ of \textsc{3-SAT-$B$}.
Each copy is on a~pairwise disjoint set of variables.
We thus create an instance $\varphi$ of \textsc{3-SAT-$B$} on $(\gamma+1)n$ variables, such that it takes $2^{\Omega(n)}$ time to tell whether all $m$ clauses of $\varphi$ are satisfiable, or at~most $m-(\gamma+1)$ clauses are, unless the ETH fails.

For $G(\varphi)$, which has $7m+2\gamma(\gamma+1) n \leqslant 7B(\gamma+1)n+2\gamma(\gamma+1) n=O(n)$ vertices, this translates into distinguishing whether its treewidth is at~most $\gamma(\gamma+1) n + 6m + \gamma - 1$ or at~least $\gamma(\gamma+1) n + 7m - (m-(\gamma+1)) - 1 = \gamma(\gamma+1) n + 6m + \gamma$, implying~\cref{thm:eth-lb}.

\medskip

\textbf{Proof of~\cref{thm:eth-inapprox}.}
Building upon the polynomial-time inapproximability of \textsc{3-SAT} shown by H{\aa}stad~\cite{Hastad01}, Moshkovitz and Raz proved that, unless the ETH fails, approximating $n$-variable \textsc{3-SAT} within ratio greater than $\frac{7}{8}$ requires time $2^{n^{1-o(1)}}$~\cite{Moshkovitz10}.
Recently, Bafna, Minzer, and Vyas improved the time lower bound to $2^{\Omega(n/\log^c n)}$ for some $c>0$~\cite{Bafna25}.

\begin{theorem}[Corollary 1.3 in \cite{Bafna25}]\label{thm:bmv}
  For any $r \in (\frac{7}{8},1)$, there is a~constant~$c := c(r)$ such that any \mbox{$r$-approx}ima\-tion algorithm for $m$-clause \textsc{3-SAT-$\log^c m$} requires $2^{\Omega(m / \log^c m)}$ time, unless the ETH fails.
\end{theorem}

Fix, say, $r := 0.88 > \frac{7}{8}$.
Let $\varphi$ be any $n$-variable $m$-clause 3-CNF formula with at~most~$\log^c m$ occurrences per variable, for $c := c(r)$ satisfying~\cref{thm:bmv}.
We slightly tune our reduction.
Instead of having a~shared value $\gamma$, we set $\gamma_i$ to be four times the number of occurrences of variable $x_i$ in~$\varphi$.
Let $\gamma$ be the average of the values $\gamma_i$, hence $\gamma n = \sum_{i \in [n]} \gamma_i \leqslant 12m$.
Note that the built graph $G(\varphi)$ has $O(m)$ vertices, and for every $i \in [n]$, $\gamma_i \leqslant \log^c m$.
\Cref{thm:bmv} creates a~gap for $G(\varphi)$ between the treewidth upper bound of $\gamma n + 6m + \log^c m - 1$ (when $\varphi$ is satisfiable) and the treewidth lower bound of $\gamma n + 6.12m - 1$ (when $\varphi$ is unsatisfiable).
For sufficiently large $m$, this makes an approximability gap of $\delta > 1$ (for some $\delta$ depending on~$c$) as $m \geqslant \gamma n/12$, which implies~\cref{thm:eth-inapprox}.

\end{document}